\documentclass[a4paper,UKenglish,cleveref,autoref,thm-restate]{lipics-v2021}

\usepackage{amsmath, amssymb}
\usepackage{complexity}

\hideLIPIcs

\title{Approximating the Chv\'atal--Gomory Closure of Capacity-Bounded Min-Closed Systems}
\titlerunning{Approximating the CG Closure}

\author{Stefano {Huber}}{Faculty of Informatics, Università della Svizzera italiana, CH-6962 Lugano, Switzerland \and
Istituto Dalle Molle di studi sull’intelligenza artificiale (IDSIA USI-SUPSI), CH-6962 Lugano,
Switzerland \and \url{https://steber97.github.io/}}{stefano.huber@usi.ch}{https://orcid.org/0009-0001-7592-6401}{}

\author{Monaldo {Mastrolilli}}{Scuola universitaria professionale della Svizzera italiana, CH-6962 Lugano, Switzerland \and Istituto Dalle Molle di studi sull’intelligenza artificiale (IDSIA USI-SUPSI), CH-6962 Lugano, Switzerland}{monaldo.mastrolilli@supsi.ch}{https://orcid.org/0000-0002-2948-9749}{}

\authorrunning{S. Huber and M. Mastrolilli}

\Copyright{Stefano Huber and Monaldo Mastrolilli}

\ccsdesc[100]{Theory of computation~Approximation algorithms analysis}
\ccsdesc[100]{Mathematics of computing~Combinatorial optimization}
\ccsdesc[100]{Theory of computation~Integer programming}
\ccsdesc[100]{Theory of computation~Design and analysis of algorithms}

\keywords{Sum of squares, Chv\'atal--Gomory closure, min-closed systems, PTAS}

\funding{The authors were supported by the Swiss National Science Foundation projects
no.\\~200021\_207429/1 ``Ideal Membership Problems and the Bit Complexity
of Sum of Squares Proofs'' and no.~200021\_212929/1 ``Computational
methods for integrality gaps analysis''.}

\relatedversion{A short version of this paper appears in the proceedings of ISAAC 2026.}

\EventEditors{John Q. Open and Joan R. Access}
\EventNoEds{2}
\EventLongTitle{42nd Conference on Very Important Topics (CVIT 2016)}
\EventShortTitle{CVIT 2016}
\EventAcronym{CVIT}
\EventYear{2016}
\EventDate{December 24--27, 2016}
\EventLocation{Little Whinging, United Kingdom}
\EventLogo{}
\SeriesVolume{42}
\ArticleNo{23}

\newcommand{\sos}{\textnormal{SoS}}

\theoremstyle{definition}
\newtheorem{assumption}[theorem]{Assumption}
\newtheorem{question}[theorem]{Question}

\begin{document}
\nolinenumbers

\maketitle

\begin{abstract}
Optimizing over the $\{0, 1/2\}$ rank-1 Chv\'atal--Gomory (CG) closure of a binary integer linear program is $\NP$-hard. While polynomial-time approximation schemes (PTAS) are established for monotone packing and covering formulations, extending these guarantees to mixed-sign variants remains an open challenge. In this paper, we study the approximability of the CG closure for \emph{$k$-slack bounded (capacity-bounded) min-closed systems}, a generalized packing formulation with mixed-sign constraints that extends weighted Boolean Horn logic. In these systems, a constant upper bound $k \ge 1$ bounds the ratio between the negative penalty coefficient $q$ and the right-hand side capacity $b$ of every constraint, enforcing $q \le k \cdot b$. We prove that a constant-degree sum-of-squares relaxation yields a PTAS for maximizing linear objectives over the first CG closure generated by multipliers in $\{0\}\cup[\tfrac1f,1]$, for any constant integer $f \ge 2$. This closure is contained in the $\{0, 1/2\}$ rank-1 CG closure.
\end{abstract}

\section{Introduction}\label{sec:introduction}

Binary integer linear programming (ILP) is a fundamental framework for modeling discrete optimization problems in theoretical computer science and operations research. A classical geometric approach to solving binary ILPs is based on cutting-plane methods, which iteratively refine the continuous linear relaxation of a problem in order to approximate its integer hull. Among these methods, the Chv\'atal--Gomory (CG) procedure is one of the most fundamental and widely used cut-generating techniques~\cite{Conforti:2014}. The procedure strengthens a fractional relaxation by taking conic combinations of valid inequalities and rounding down the resulting right-hand sides. The general intuition being, if $x_1 + x_2 \leq 1.5$ is a constraint of an ILP, it may well be rewritten as $x_1 + x_2 \leq 1$.

\begin{definition}[Chv\'atal--Gomory Closure]\label{def:CG_closure}
Let $I$ be an instance given by a linear system $Ax \le b$, with
$A \in \mathbb{Z}^{m\times n}$, $b \in \mathbb{Z}^m$, $x\in\{0,1\}^n$, and let $P(I) := \{x\in [0,1]^n : Ax \le b \}$ denote its continuous relaxation.
For any multiplier vector $\lambda \in \mathbb{R}_{\ge 0}^m$ such that $\lambda^\top A \in \mathbb{Z}^n$, the inequality $(\lambda^\top A)x \le \lfloor \lambda^\top b \rfloor$ is called a \emph{rank-1 Chv\'atal--Gomory (CG) cut} for $P(I)$. 

The intersection of $P(I)$ with all rank-1 CG cuts is called the \emph{first CG closure} and is denoted by $P^{(1)}(I)$. Recursively, for $t \ge 1$, the $t$-th CG closure is defined as
$P^{(t)}(I) = (P^{(t-1)}(I))^{(1)}$, with $P^{(0)}(I) = P(I)$. When the instance $I$ is fixed and clear from context, we write $P^{(t)}$ for $P^{(t)}(I)$. Notice that, for all $t\geq 0$, $P^{(t)}(I)$ is a polyhedron~\cite{Conforti:2014}. 

Throughout, the variable bounds $0 \le x_i \le 1$, $i \in [n]$, can be taken to be part of the system $Ax \le b$, so that they too may serve as generators of CG cuts. This is what lets fractional constraint-multipliers, such as $\lambda \in \{0,\tfrac12\}^m$, produce \emph{integral} aggregate coefficients: a half-integral left-hand side is integralized by combining it with the bounds $-x_i\leq 0$ before the right-hand side is rounded down.
\end{definition}

To explicitly capture the underlying mixed-sign structure, we write a rank-1 conic combination with integer left-hand side before the CG rounding step as
\begin{equation}
    \sum_{i \in S} a_i x_i - \sum_{k \in K} q_k y_k \le b,
    \label{eq:generic_f}
\end{equation}
where $a_i,q_k \in \mathbb{Z}_{\ge 1}$, $b \in \mathbb{R}_{\ge 0}$, and $S$ and $K$ are disjoint subsets of $[n]:=\{1, \dots, n\}$, the set of variable indices. Notice that we write $y_k := x_k$ for $k\in K$ to highlight the variables with negative coefficient. The CG cut is obtained by replacing $b$ with $\lfloor b \rfloor$.

By iteratively applying this rounding operator, the CG procedure generates progressively tighter polyhedral relaxations. Many of the most effective classes of valid inequalities in polyhedral combinatorics arise as low-rank CG cuts, including the odd-cycle inequalities for the stable set polytope, blossom inequalities for the matching polytope, simple M\"obius ladder inequalities for the acyclic subdigraph polytope, and simple comb inequalities for the symmetric traveling salesman polytope (see, e.g.,~\cite{Conforti:2014}). Interestingly, all these cuts can be derived as in \cref{def:CG_closure} with $\lambda \in \{0, 1/2\}^m$.

The CG rank required to reach the integer hull has been extensively studied. For polytopes contained in the unit cube, Eisenbrand and Schulz~\cite{EisenbrandS03} proved that $O(n^2 \log n)$ rounds suffice, and this bound was shown to be nearly tight by Rothvo\ss{} and Sanità~\cite{RothvossS17}.

The sum-of-squares ($\sos$) hierarchy provides one of the most powerful general-purpose frameworks for constructing tight convex relaxations in combinatorial optimization. Interestingly, the Chv\'atal--Gomory operator can capture strong linear relaxations that escape bounded-degree $\sos$ relaxations (see, e.g., \cite{FlemingKothariPitassi19}). In particular, the $\sos$-hard instances constructed in~\cite{KurpiszLM17}, which require high-degree $\sos$ certificates, are completely described by the first CG closure $P^{(1)}$. Thus, instances that are difficult for low-degree $\sos$ proofs can be remarkably easy from the perspective of CG cuts. However, this expressive power comes at a substantial computational cost. 

Already the $\{0,\tfrac12\}$ rank-1 CG closure---the CG closure generated by the defining constraints together with the box inequalities $0 \leq x_i \leq 1$, using multipliers $\lambda$ restricted to $\{0,1/2\}$ and followed by the standard CG rounding---is computationally hard: deciding membership in it is $\coNP$-complete, already for binary programs with $0/1$ constraint matrices, so optimizing a linear objective over it is $\NP$-hard~\cite{LetchfordPS11}. We denote this $\{0,\tfrac12\}$ closure as $P^{(1)}_{\{0,1/2\}}$. 

To overcome this intrinsic computational hardness, algorithmic research has increasingly focused on \textit{polynomial-time approximation schemes} (PTAS). In our setting, it is convenient to interpret approximation guarantees in terms of upper bounds induced by Chv\'atal--Gomory relaxations. Let $\mathcal{F}_{A,b} = \{x \in \{0,1\}^n : Ax \le b\}$ denote the set of feasible integer solutions. For an instance $I$ of a maximization problem, let
\[
\mathrm{OPT}(I) = \max\{c^\top x : x \in \mathcal{F}_{A,b}\}
\]
denote the optimal value over the integer-feasible solutions for $c\in\mathbb{R}^n_{\ge 0}$, which is $\NP$-hard to compute. For each $t \ge 0$, the maximization problem over the CG closure $P^{(t)}(I)$ defines a relaxation value
\[
\mathrm{OPT}^{(t)}(I) = \max\{c^\top x : x \in P^{(t)}(I)\},
\]
which induces a hierarchy of increasingly tighter upper bounds:
\[
\mathrm{OPT}(I) \;\le\; \mathrm{OPT}^{(t)}(I) \;\le\; \mathrm{OPT}^{(t-1)}(I) \;\le\; \cdots \;\le\; \mathrm{OPT}^{(0)}(I),
\]
where $\mathrm{OPT}^{(0)}(I)$ is the value of the LP relaxation of the original binary ILP.

We say that an algorithm is a PTAS for the $t$-th CG closure if, for every fixed precision $\varepsilon > 0$ and fixed rank $t\in\mathbb{Z}_{>0}$, the algorithm runs in time polynomial in the input size and outputs a value $U(I)$ satisfying
\[
\mathrm{OPT}(I) \;\le\; U(I) \;\le\; (1+\varepsilon)\,\mathrm{OPT}^{(t)}(I).
\]
Equivalently, for each fixed precision $\varepsilon$, the algorithm computes in polynomial time a $(1+\varepsilon)$-approximate upper bound for the $t$-th CG relaxation. A symmetric formulation applies to minimization problems even though, in this paper, we discuss only maximization problems.

A major milestone for covering problems was achieved by Bienstock and Zuckerberg~\cite{BienstockZ06}, who established a PTAS for optimizing over the $t$-th CG closure for any constant rank $t=O(1)$. Their approach leverages a specialized lift-and-project framework driven by a fundamental structural property: constraint violations are entirely governed by inclusion-minimal infeasible subsets.

Complementing these results for covering, Mastrolilli~\cite{MastrolilliHsos} established the analogous PTAS specifically for packing polyhedra. This is achieved by employing the standard $\sos$ hierarchy for 0/1 problems. Crucially, for pure packing formulations, a constant-degree $\sos$ relaxation approximates the optimization over the $t$-th CG closure within an arbitrarily small error for any constant rank $t$, matching the approximation guarantees available for the covering counterpart.

In the regimes considered so far, covering and packing problems admit a monotone structure, where constraints can be written as $\sum_i a_i x_i \le b$ with consistent sign patterns (nonnegative for packing and nonpositive for covering). This monotonicity underlies all aggregation and rounding techniques used above.

In contrast, general integer programs may contain mixed-sign constraints, where this structure breaks down and the approximation techniques above no longer apply. This highlights a fundamental gap in the literature: while packing and covering admit robust approximation schemes under monotonicity, for mixed-sign integer programs no general approximation framework is currently known. The main obstacle is the loss of monotonicity, which invalidates aggregation and rounding arguments. Characterizing which mixed-sign classes still admit polynomial-time approximation guarantees remains an open problem.

\section{Problem Setting and Overview of Results}

A generalization of packing problems is given by \emph{min-closed systems} (see, e.g.,~\cite{barto_et_al:DFU:2017:6959} and \Cref{def:min-closed} below). In this paper, we analyze the approximability of the CG closure for a class of mixed-sign linear inequality systems whose Boolean feasible regions exhibit a min-closed structure.

We prove that under certain conditions (see \Cref{ass:k_slack}), a constant-degree $\sos$ relaxation yields a PTAS for maximizing linear objectives over the first $\{0\}\cup[\tfrac1f,1]$-CG closure. This closure is defined by rank-1 conic combinations whose non-zero Farkas multipliers on the structural constraints are in $[1/f,1]$ for a fixed constant integer $f \ge 2$. More details are given in the rest of the section.

\subsection{Min-Closed Systems and Valid Inequalities}

We consider a Boolean feasible region $\mathcal{F}_{A,b} = \{x \in \{0,1\}^n : Ax \le b\}$, where the input matrix $A \in \mathbb{Z}^{m \times n}$ and the vector $b \in \mathbb{Z}^m$ have integer components. 

\begin{definition}[Min-Closed Set and Constraint]\label{def:min-closed}
A set $T \subseteq \{0,1\}^n$ is min-closed if for any $x, y \in T$, their coordinate-wise minimum $z = x \wedge y$ (where $z_i = \min(x_i, y_i)$) also belongs to $T$. A linear constraint is min-closed if its Boolean feasible region is a min-closed set.
\end{definition}

Packing constraints are min-closed, but min-closed systems form a strictly broader class. Min-closed constraints have an exact logical characterization in terms of Horn clauses.

\begin{definition}[Boolean Horn Clause and Linear Representation]
Let $X = \{x_1, \dots, x_n\}$ be a set of Boolean variables, and let $y \in \{0,1\}$ be a designated Boolean variable not belonging to $X$. A \emph{Boolean Horn clause} over $X \cup \{y\}$ is a disjunction of literals containing at most one positive literal, which, when present, is $y$. Equivalently, it can be written as an implication of the form
\[
\bigwedge_{i \in P} x_i \;\Rightarrow\; y,
\]
where $P$ (the \emph{premise}) is the set of indices of the negative literals, and $y$ is the unique positive literal in the \emph{conclusion} (it may be absent, in which case the clause is purely negative).

This implication admits a linear representation:
\begin{equation} \label{eq:horn_linear}
    \sum_{i \in P} (1 - x_i) + y \ge 1
    \quad \Longleftrightarrow \quad
    \sum_{i \in P} x_i - y \le |P| - 1.
\end{equation}
We call $|P|$ the \emph{premise width} of the Horn clause.
\end{definition}

\begin{definition}[Generalized Horn Constraint]\label{def:gen_horn}
Let $X = \{x_1, \dots, x_n\}$ be a set of Boolean variables, and let $y = x_j$ for a designated Boolean variable. A constraint $\mathcal{C}$ is defined as a \emph{generalized Horn constraint} if it can be written as:
\begin{equation} \label{eq:constraint_A}
    \sum_{i=1}^n a_i x_i - q y \le b
\end{equation}
where the coefficients $(a_1, \dots, a_n), q \in \mathbb{Z}_{\ge 0}$, $a_j=0$, and the threshold $b \in \mathbb{Z}_{\geq 0}$.
\end{definition}

We call a Boolean Horn clause \emph{composite} if its premise width satisfies $|P| \ge 2$: the special case $a_i = 1$ for $i\in P$, $a_i=0$ for $i\notin P$, $q=1$, $b = |P|-1 \ge 1$ of \Cref{eq:constraint_A} is captured by composite Horn clauses. The unitary implications $x_i \Rightarrow y$ (which correspond to $b = 0$ under \Cref{eq:horn_linear}) and the unconditional assertions lie outside the composite class: these constraints, as we discuss later, are not handled by our algorithm.

A generalized Horn constraint can be seen as a soft version of a packing constraint, where the negative term $q y$ acts as a slack term that is activated whenever the positive variables violate the capacity $b$. The following proposition establishes that generalized Horn constraints are min-closed.

\begin{proposition}[Min-Closed Property]\label{prop:min_closed}
Every generalized Horn constraint $\mathcal{C}$ is min-closed. That is, for any two satisfying assignments
$x^1, x^2 \in \{0,1\}^{n}$, their coordinate-wise minimum $x^1 \wedge x^2$ (equivalently, their bitwise AND) is also a satisfying assignment of $\mathcal{C}$.
\end{proposition}
\begin{proof}
Let $y^1 := x^1_j$ (the negative variable) and $y^2$ respectively, $z_i = \min(x^1_i, x^2_i)$ for $i\in\{1,\dots,n\}$, and $w = z_j = \min(y^1, y^2)$.
Pick $\ell \in \{1, 2\}$ such that $y^\ell = w$ (the assignment attaining the minimum on the
negative variable). Since $z_i \le x^\ell_i$ for all $i$ and every $a_i \ge 0$,
\[
    \sum_{i=1}^n a_i z_i - q\,w \;\le\; \sum_{i=1}^n a_i x^\ell_i - q\,y^\ell \;\le\; b,
\]
where the last inequality holds because $x^\ell$ satisfies $\mathcal{C}$. Hence
$x^1 \wedge x^2$ satisfies $\mathcal{C}$.
\end{proof}

\begin{remark}[The Representation Gap]\label{rem:representation_gap}
Every min-closed relation is equivalent to a conjunction of Boolean Horn clauses~\cite{JeavonsC95}. However, a single linear min-closed constraint can encode a relation whose Horn-clause representation requires exponentially many clauses (see \Cref{ex:exponential_gap} in \Cref{app:exponential_gap}, where a single generalized Horn constraint is equivalently represented by an exponential amount of Horn clauses). This exponential blow-up implies that an efficient approximation scheme cannot expand a min-closed system into its Horn-clause representation, but rather, it should work with the generalized Horn constraints as they are given. This is what our algorithm does.
\end{remark}

The central question investigated in this paper is the following:

\begin{question}[Main Algorithmic Goal]\label{Q:PTAS_MIN}
Does there exist a PTAS for maximizing linear objectives over the first Chv\'atal--Gomory closure $P^{(1)}$ when the feasible region is defined by arbitrary min-closed linear constraints?
\end{question}

Since resolving \Cref{Q:PTAS_MIN} in its full generality remains an open challenge, we focus on a natural structural subclass of min-closed systems, and we restrict the non-zero multipliers of the first CG closure to an interval that generalizes the classical discrete $\{0, \tfrac12\}$ cuts.

Throughout we write $\Lambda_f := \{0\} \cup \left[1/f, 1\right]$ for this multiplier set, and $P^{(1)}_{\Lambda_f}$ for the corresponding first $\Lambda_f$-CG closure where the Farkas multipliers for the structural constraints live in $\Lambda_f$, whereas the Farkas multipliers for the lower bounds $-x_i\leq 0$ live in $\mathbb{R}_{\ge 0}$ (as explained in \Cref{thm:cg-bound-redundancy}, they can be taken without loss of generality in $[0,1)$). Notice that $P^{(1)}_{\Lambda_f} \subseteq P^{(1)}_{\{0,1/2\}}$, so that a $(1+\varepsilon)$-approximation on $P^{(1)}_{\Lambda_f}$ is an upper bound on $\mathrm{OPT}(I)$ at least as good as a $(1+\varepsilon)$-approximation on $P^{(1)}_{\{0,1/2\}}$.

Algebraically, enforcing $\lambda_\ell \ge \tfrac1f$ ensures that the non-zero weights in the linear combination cannot be arbitrarily small, preventing uncontrolled error propagation during the Chv\'atal--Gomory rounding step. This property allows our sum-of-squares relaxation to preserve a constant degree.

\subsection{Our Results}
\label{sec:our_results}

We identify a regime that renders the mixed-sign nature of generalized Horn systems algorithmically tractable. To this end, we introduce the class of \emph{$k$-slack bounded} constraints.

\begin{assumption}[$k$-Slack Boundedness]\label{ass:k_slack}
There exists a universal constant $k \in \mathbb{Z}_{\ge 1}$ such that every generalized Horn constraint (\Cref{def:gen_horn})
\[
    \sum_{i=1}^n a_i x_i - q\,y \le b,
    \qquad
    a_1, \dots, a_n,\, q,\, b \in \mathbb{Z}_{\ge 0},
\]
defining the system satisfies $q \le k \cdot b$.
\end{assumption}

Exact optimization over $P_{\{0,1/2\}}^{(1)}$ is $\NP$-hard even for monotone packing formulations ($A \in \{0,1\}^{m \times n}$, $q=0$, $b=\mathbf{1}$)~\cite{LetchfordPS11}. Since this monotone case is included in the $k$-slack bounded regime, optimizing over $P^{(1)}_{\Lambda_f}$ provides a bound which is at least as tight as an $\NP$-hard relaxation: this justifies the search for an approximation algorithm.

Our main algorithmic contribution establishes that a constant-degree $\sos$ relaxation provides a $(1+\varepsilon)$-approximation for the first CG closure generated by multipliers in $\Lambda_f$. Notice that $P^{(1)}_{\Lambda_f}$ is a polytope: by \Cref{thm:cg-bound-redundancy} it suffices to intersect the \emph{tight} cuts, whose multipliers on the lower bounds $-x_i \le 0$ lie in $[0,1)$; since $\Lambda_f \subseteq [0,1]$ as well, every multiplier is at most $1$. As $A$ and $b$ are integral, the left-hand side of such a cut is an integer vector bounded by the column sums $\sum_\ell |A_{\ell i}|$, and its right-hand side $\lfloor \lambda^\top b \rfloor$ is bounded too, so both range over finite sets. Hence only finitely many distinct cuts arise, and $P^{(1)}_{\Lambda_f}$ is the intersection of $P(I) \subseteq [0,1]^n$ with finitely many halfspaces.

We state here our main contribution: the $\sos$ hierarchy at constant level yields a $(1+\varepsilon)$-approximation of the first $\Lambda_f$-CG closure for $k$-slack systems of generalized Horn constraints.

\begin{theorem}[Constant-Level $\sos$ Approximation of the $\Lambda_f$-CG Closure]\label{thm:main_ptas}
Let $I$ be an instance of maximizing a linear function $c^\top x$, $c\in\mathbb{R}^{n}_{\ge 0}$, subject to a system of generalized Horn constraints $\mathcal{F}_{A,b} = \{x \in \{0,1\}^n : Ax \le b\}$ satisfying \Cref{ass:k_slack} for a constant $k \ge 1$; fix an integer $f \ge 2$, and let $P^{(1)}_{\Lambda_f}$ be its first $\Lambda_f$-CG closure. For any fixed $\varepsilon > 0$, there exists a constant $D = O\bigl((k+f)/\varepsilon\bigr)$ such that the optimal value $\mathrm{OPT}_D(I)$ of the $D$-$\sos$ relaxation of $\mathcal{F}_{A,b}$ satisfies:
\[
    \mathrm{OPT}(I) \;\le\; \mathrm{OPT}_D(I) \;\le\; (1+\varepsilon)\,\mathrm{OPT}^{(1)}_{\Lambda_f}(I).
\]
\end{theorem}

The proof of \Cref{thm:main_ptas} is the content of \Cref{sec:proof_main_thm}. Notice that $\mathrm{OPT}_D(I)$ might not be computable in polynomial time due to the bit-complexity issues raised by \cite{ODonnell17}, even approximately. The following remark shows that, using \Cref{thm:main_ptas} and~\cite{MastrolilliSoSTractability}, we can compute in polynomial time a rational point $\hat x\in P^{(1)}_{\Lambda_f}$ and a rational upper bound $U$ for $\mathrm{OPT}(I)$.

\begin{remark}[Polynomial-Time Implementation]\label{rem:ptime-sos}
Under the assumptions of \Cref{thm:main_ptas}, for every fixed rational $\varepsilon>0$ and fixed $k,f$, a polynomial-time algorithm returns a rational point $\hat x\in P^{(1)}_{\Lambda_f}$ and a rational number $U$ satisfying
\[
  \mathrm{OPT}(I)
  \le U
  \le (1+\varepsilon)c^\top\hat x
  \le (1+\varepsilon)\mathrm{OPT}^{(1)}_{\Lambda_f}(I).
\]
Thus $U$ is an upper bound on the integer optimum, whereas $\hat x$ is feasible for $P^{(1)}_{\Lambda_f}$ and has objective value at least $\mathrm{OPT}(I)/(1+\varepsilon)$.
The number $U$ need not equal $c^\top\hat x$.

To obtain these outputs, let $Q\ge1$ be an integer such that $Qc\in\mathbb Z^n$; the product of the denominators of the rational objective coefficients is a suitable choice with polynomial encoding length. Set
\[
  d=1+\left\lceil\frac{2}{\varepsilon}\right\rceil,
  \qquad \phi=1-\frac1d,
  \qquad \delta=\frac{\varepsilon}{2Q(1+\varepsilon)},
\]
and take $D=\bigl(2(k+f)+1\bigr)d+2(k+f)-1$ as in \Cref{cor:scaled_in_Pt}. A $D$-$\sos$ certificate (\Cref{def:sos_system}) squares polynomials of degree at most $D$ and multiplies them by the linear axioms, so its total degree is at most $2D+1$. We therefore apply the min-closed pseudoexpectation optimization result of Mastrolilli~\cite[Theorem~4.1 and Corollary~5.8]{MastrolilliSoSTractability} at total degree $r=2D+2$, whose running time is polynomial in the input length and in $\log(1/\delta)$. Every row is min-closed, as that result requires: the structural rows by \Cref{prop:min_closed}, the box rows trivially; note that this condition, unlike \Cref{ass:k_slack}, also holds for the non-negativity bounds.

That result optimizes, to additive accuracy $\delta$, over an augmented moment body $K_r\subseteq S_r$ containing the evaluation vectors of all integer-feasible points. Here $S_r$ denotes the feasible region of the standard total-degree-$r$ $\sos$ moment relaxation, and $K_r\subseteq S_r$ is obtained by additionally requiring the pseudoexpectation to vanish on every polynomial of degree at most $r$ that vanishes on all integer-feasible points. It returns a rational pseudoexpectation  satisfying all degree-$r$ constraints; since $r = 2D+2$, this pseudoexpectation is in particular a solution of the $D$-$\sos$ relaxation of $\mathcal{F}_{A,b}$ in the sense of \Cref{rem:duality}. Its first-moment vector $x$ consequently satisfies $u:=c^\top x\ge\mathrm{OPT}(I)-\delta$.
Notice that $u\ge 0$ since both $x, c\ge 0$. By \Cref{cor:scaled_in_Pt}, $\hat x:=\phi x$ belongs to $P^{(1)}_{\Lambda_f}$. Let $U:=\frac{\left\lfloor Q(u+\delta)\right\rfloor}{Q}$.

Since $Q\,\mathrm{OPT}(I)$ is an integer, the lower bound on $u$ implies $U\ge\mathrm{OPT}(I)\ge0$. As $QU=\lfloor Q(u+\delta)\rfloor$ is a nonnegative integer, either $U=0$---in which case $U\le(1+\varepsilon)c^\top\hat x$ because $c^\top\hat x=\phi u\ge0$---or $U\ge1/Q$. In the latter case, $U\le u+\delta$ gives $u\ge U-\delta\ge(1-Q\delta)U$.
Since $c^\top\hat x=\phi u$, this rearranges to
\[
  U\;\le\;\frac{c^\top\hat x}{\phi\,(1-Q\delta)}\;\le\;(1+\varepsilon)\,c^\top\hat x ,
\]
where the second inequality uses $c^\top\hat x\ge0$ together with $\phi^{-1}\le1+\varepsilon/2$ and $Q\delta=\varepsilon/[2(1+\varepsilon)]$, whose combination is exactly $\bigl(1+\tfrac{\varepsilon}{2}\bigr)\bigl(1-\tfrac{\varepsilon}{2(1+\varepsilon)}\bigr)^{-1}=1+\varepsilon$. The last inequality of the chain follows from $\hat x\in P^{(1)}_{\Lambda_f}$. For fixed $k,f,\varepsilon$, the degree is constant and the encoding lengths of $Q$ and $\delta$ are polynomial in the input length, so both outputs are computable in polynomial time. No additive error remains in the stated guarantees, and no feasibility error is introduced.
\end{remark}

\subsection{Paper Organization}
\label{sec:organization}

The technical derivation of our results proceeds as follows. \Cref{sec:violation_threshold} introduces the \emph{violation threshold} (\Cref{def:violation_threshold}) and establishes our core algebraic tool (\Cref{lem:sos_exactness}), proving that constant-degree $\sos$ proofs certify Chv\'atal--Gomory cuts with bounded capacities exactly.

In \Cref{sec:proof_main_thm}, we characterize the structure of non-dominated cuts in $P^{(1)}_{\Lambda_f}$ and bound their negative weight accumulation. We then combine this structural bound, the algebraic exactness from \Cref{sec:violation_threshold}, and a uniform geometric scaling argument to prove \Cref{thm:main_ptas}.

Finally, in \Cref{sec:discussion} we raise some possible future directions in order to establish an answer to \Cref{Q:PTAS_MIN}.

Due to space constraints, all omitted proofs, along with the requisite preliminaries on the $\sos$ proof system (\Cref{sec:preliminaries}), are deferred to the appendix.

\section{The Violation Threshold and \texorpdfstring{$\sos$}{SoS} Exactness}\label{sec:violation_threshold}

Throughout the proof of \Cref{thm:main_ptas} we handle each cut by a case analysis on a single integer parameter, the \emph{violation threshold} $\tau$ (\Cref{def:violation_threshold}): the largest number of positive variables that can be set to $1$ before the cut is violated. Small and large thresholds are handled by different mechanisms: if $\tau$ is bounded by a constant, then a constant-degree $\sos$ relaxation certifies the integer rounding of the Chv\'atal--Gomory cut exactly (\Cref{lem:sos_exactness}). If, instead, $\tau$ is large, the error introduced by rounding down the right-hand side is a negligible fraction of the capacity and we absorb it by scaling the solution by a uniform geometric factor.

\begin{definition}[Violation Threshold $\tau$]\label{def:violation_threshold}
Consider a valid fractional inequality $\sum_{i \in S} a_i x_i - \sum_{k \in K} q_k y_k \le b$ as defined in \cref{eq:generic_f}. Assume w.l.o.g. that the positive coefficients are sorted in non-decreasing order, $a_1 \le a_2 \le \dots \le a_{|S|}$. The violation threshold $\tau$ is defined as:
\begin{equation}
    \tau := \max\left\{i \in \{0,1,\dots,|S|\}\; :\; \sum_{j=1}^{i} a_j \le b + \sum_{k \in K} q_k \right\}
\end{equation}
\end{definition}

When $\tau = O(1)$, the size of any minimal violating subset of positive variables is constant. This structural property allows a constant-level $\sos$ certificate for the Chv\'atal--Gomory rounding step.

\begin{lemma}[$\sos$ Exactness for Bounded Thresholds]\label{lem:sos_exactness}
Let $\sum_{i \in S} a_i x_i - \sum_{k \in K} q_k y_k \le b$ be a fractional inequality as in \cref{eq:generic_f}, obtained as a nonnegative combination of the constraints of $P$, and let its violation threshold be at most $\tau$, with $\tau + |K| \ge 1$. Then its Chv\'atal--Gomory rounding $\lfloor b \rfloor - \sum_{i \in S} a_i x_i + \sum_{k \in K} q_k y_k \ge 0$ admits a $(\tau + |K| + 1)$-$\sos$ certificate (see \Cref{def:sos_system}) from the constraints of $P$.
\end{lemma}

\Cref{lem:sos_exactness} follows from the following result of~\cite{MastrolilliHsos} on packing constraints, which in turn can be derived from the Decomposition Theorem of~\cite{KarlinMN11}.

\begin{lemma}[Lemma 6.1 in~\cite{MastrolilliHsos}]\label{lem:decomposition_theorem}
    Consider any packing problem instance given by a matrix $A \in \mathbb{Z}^{m \times n}_+$ and a vector $b \in \mathbb{R}^{m}_+$. Let $\pi = O(1)$ be a fixed positive integer. Then, for any valid inequality $a_0 - a^\top x \ge 0$ for $\mathcal{F}_{A,b} := \{ x \in \{0,1\}^n: Ax\le b\}$ with violation threshold at most $\pi$\footnote{For packing inequalities, the violation threshold is defined as in \Cref{def:violation_threshold} using $K=\emptyset$. In \cite{MastrolilliHsos} it is called ``pitch''.}, there exists a $(\pi+1)$-$\sos$ certificate of non-negativity from the constraints $b - Ax \ge 0$.\footnote{The Lemma as stated in~\cite{MastrolilliHsos} gets the $(\pi + 1)$-$\sos$ certificate as the first step of the proof.}
\end{lemma}

The proof of \Cref{lem:sos_exactness} is deferred to \Cref{sect:proof_lem:sos_exactness}.

\section{Proof of \texorpdfstring{\Cref{thm:main_ptas}}{Proof of the main Theorem}}\label{sec:proof_main_thm}

This section proves \Cref{thm:main_ptas} by showing that the scaling $\hat{x} = (1 - \frac{1}{d})x^\ast$ of a pseudo-expectation vector $x^\ast$ from a constant-level $\sos$ relaxation satisfies all valid rank-1 $\Lambda_f$-CG cuts. The proof is organized as follows:

\begin{enumerate}
    \item Structural Analysis of CG Cuts: in \Cref{sec:structural_analysis_cg_cuts}, we identify the class of non-dominated rank-1 cuts. We prove that restricting the non-negativity bounds' multipliers to $[0,1)$ is without loss of generality (\Cref{thm:cg-bound-redundancy}) and that the total negative weight in any non-dominated cut is bounded by a linear function of its capacity (\Cref{thm:restricted_weight_bound}).
    
    \item Satisfiability under Geometric Scaling: in \Cref{sec:satisfiability_under_geometric_scaling}, we analyze the validity of the scaled solution $\hat{x} = (1-1/d)x^\ast$ (for a carefully chosen constant $d$) by distinguishing two cases based on the cut capacity $\tilde{b}$:
    \begin{itemize}
        \item \emph{Low Capacity ($\tilde{b} \le d$):} we show that these cuts have bounded violation threshold. Consequently, the $D$-$\sos$ relaxation satisfies these cuts via the algebraic certificate provided by \Cref{lem:sos_exactness} and \Cref{rem:duality} (\Cref{lem:sos_degree_all_ranks}).
        \item \emph{High Capacity ($\tilde{b} > d$):} we prove that uniform geometric scaling is sufficient to satisfy the rounded capacity constraint, irrespective of the violation threshold of the cut (\Cref{lem:scaling_all_ranks}).
    \end{itemize}
    
    \item Approximation Guarantee: finally, in \Cref{sec:approx_guarantee}, we show that the scaled solution $\hat{x}$ lies in $P^{(1)}_{\Lambda_f}$ (\Cref{cor:scaled_in_Pt}) and we derive the $(1+\varepsilon)$-approximation ratio for the maximization objective.
\end{enumerate}
Throughout this section, we refer to CG cuts and closures over $\lambda\in\Lambda_f^m$, omitting the explicit reference to the multiplier set when clear from context.

For the entire section, we will assume that $I$ is an instance of a maximization problem of the linear function $c^\top x$ ($c\in\mathbb{R}^{n}_{\ge 0}$) subject to a set of linear constraints $Ax \le b$ satisfying \Cref{ass:k_slack}, together with the box constraints $0 \le x \le 1$. We will refer to the resulting polytope $P(I)$ as $P$. Notice that the lower bounds $x_i \ge 0$ ($-x_i \leq 0$), unlike the upper bounds $x_i \leq 1$, fall outside the $k$-slack regime, yet both bounds are necessary in order to create a valid CG closure. Hence, in the following discussion, lower bounds are handled with a dedicated argument, while upper bounds are assumed to be part of $Ax \le b$.

\subsection{Structural Analysis of CG Cuts}\label{sec:structural_analysis_cg_cuts}

We start by describing what the cuts of the first CG closure look like:

\begin{remark}\label{rem:cuts_first_closure}
    
Recall $P$ is defined by a system $Ax \le b$ of $m$ constraints satisfying \Cref{ass:k_slack}. 
Each constraint $\ell \in \{1, \dots, m\}$ is a generalized Horn constraint (\Cref{def:gen_horn}) and has structural form:
\[
    \sum_{i \in S_\ell} a_{\ell, i} x_i - c_\ell x_{j_\ell} \le b_\ell
\]
where $S_\ell \subseteq \{1, \dots, n\} \setminus \{j_\ell\}$ is the set of positive variable indices for constraint $\ell$, and the coefficients satisfy $c_\ell \le k b_\ell$ (i.e., the system satisfies \Cref{ass:k_slack}). Notice that coefficients $a_{\ell,i} = 0$ for $i\notin S_\ell$ and $c_\ell$ might be zero in case the negative variable is not present in constraint $\ell$ (e.g., for box constraints $x_i\le 1$).

Consider a conic combination of the structural constraints with multipliers $\lambda_\ell \in \Lambda_f$ and the non-negativity bounds $-x_i \le 0$ with multipliers $\mu_i \ge 0$.

Let $w_i(\lambda)$ denote the aggregated structural coefficient for variable $x_i$ prior to bound inclusion:
\[
    w_i(\lambda) = \sum_{\ell : i \in S_\ell} \lambda_\ell a_{\ell, i} - \sum_{\ell : j_\ell = i} \lambda_\ell c_\ell.
\]
The resulting inequality before rounding is:
\[
    \sum_{i=1}^n \big( w_i(\lambda) - \mu_i \big) x_i \le \sum_{\ell=1}^m \lambda_\ell b_\ell
\]

The latter is a valid CG cut when $w_i(\lambda) - \mu_i$ is an integer, for every $i\in [n]$.
We define the minimal fractional absorber $\mu^{tight}_i = w_i(\lambda) - \lfloor w_i(\lambda) \rfloor$, ensuring $\mu^{tight}_i \in [0, 1)$. The cut generated using $\mu = \mu^{tight}$ is referred to as the \emph{tight cut}. 
\end{remark}

Now we prove that, when building the first CG closure, the multipliers $\mu_i$ of the constraints $-x_i\leq 0$ can be taken in $[0, 1)$. In fact, whenever some $\mu_i > \mu_i^{tight}$, we are de facto creating a dominated cut, which is redundant.

\begin{theorem}[Redundancy of Non-Minimal Multipliers for Non-Negativity Bounds]\label{thm:cg-bound-redundancy}

Using the same notation as in \Cref{rem:cuts_first_closure}, if a CG cut is generated using an arbitrary valid multiplier vector $\mu$ such that its components $\mu_i = \mu^{tight}_i + d_i$ for some nonnegative integers $d_i$ that are not all zero; then, this cut is dominated by the corresponding tight cut (using $\mu_i = \mu^{tight}_i$). Consequently, adding any integer weight $d_i \ge 1$ to the non-negativity bounds $-x_i \le 0$ is redundant in the formulation of the first CG closure $P^{(1)}_{\Lambda_f}$.

\end{theorem}

The proof of \Cref{thm:cg-bound-redundancy} is given in the appendix (\Cref{app:proof_of_thm_cg-bound-redundancy}).

\Cref{thm:cg-bound-redundancy} allows us to restrict our analysis of the first $\Lambda_f$-CG closure exclusively to non-dominated tight cuts. We now prove that this restriction preserves the structural negative-weight bound of the original system.

\begin{theorem}[Preservation of Negative Weight Bounds in the CG Closure]\label{thm:cg-bound-preservation}
    Let $P$ satisfy \Cref{ass:k_slack} with constant $k$, and let $C_{tight}$ be any non-dominated rank-1 $\Lambda_f$-CG cut for $P$, obtained with multipliers $\lambda \in \Lambda_f^m$ and $\mu \in [0,1)^n$, written in standard form as:
    \[
    \sum_{i=1}^n c^{CG}_i x_i \le \beta = \left\lfloor \sum_\ell \lambda_\ell b_\ell \right\rfloor;
    \]
    then, for any variable $x_i$ exhibiting a negative coefficient ($c^{CG}_i < 0$), its magnitude is bounded by $|c^{CG}_i| \le k( \beta + 1 )$.
    Consequently, the magnitude of the negative coefficients across all non-dominated valid inequalities in $P^{(1)}_{\Lambda_f}$ remains bounded by $O(\beta)$.
\end{theorem}

The proof of \Cref{thm:cg-bound-preservation} is in the appendix (\Cref{app:proof-cg-bound-preservation}).

\begin{theorem}[Bound on Total Negative Weight]
\label{thm:restricted_weight_bound}
For any non-dominated rank-1 $\Lambda_f$-CG cut written as $\sum_{i=1}^n c^{CG}_i x_i \le \beta$ as in \Cref{thm:cg-bound-preservation}, generated from $P$, and satisfying \Cref{ass:k_slack} with constant $k$, with multipliers $\lambda \in \Lambda_f^m$ for some constant $f\ge 2$ and $\mu \in [0,1)^n$, the total negative weight $\Delta_{cut} := \sum_{i: c^{CG}_i < 0} |c^{CG}_i|$ satisfies: $\Delta_{cut} \le (k+f)\beta + (k+f) - 1$, where $\beta := \left\lfloor \sum_\ell \lambda_\ell b_\ell \right\rfloor$.
Consequently, for fixed $k$ and $f$, $\Delta_{cut} = O(\beta)$.
\end{theorem}

The proof of \Cref{thm:restricted_weight_bound} is given in the appendix (\Cref{app:proof-restricted-weight-bound}).

\begin{remark}
The bound $\Delta_{cut} = O(\beta)$ ensures that the negative weight accumulation in our 
CG closure is bounded by the system's capacities. This ensures that the level of the $\sos$ certificate of \Cref{lem:sos_exactness} is bounded by $O(\beta)$, and independent of the number of variables $n$.
\end{remark}

\subsection{Satisfiability Under Geometric Scaling}\label{sec:satisfiability_under_geometric_scaling}

With these key ingredients in place, we are ready to prove the main \Cref{thm:main_ptas}.
The strategy is as follows: for every instance $I$ satisfying \Cref{ass:k_slack}, we pick a constant $d := \lceil 1/\ln(1+\varepsilon)\rceil + 1$ depending on the precision parameter $\varepsilon$, such that $\frac{d}{d-1} \leq 1+\varepsilon$ (see the proof of \Cref{thm:main_ptas}). Let $x^\ast$ be a solution of the $D$-$\sos$ relaxation (with $D$ as in \Cref{lem:sos_degree_all_ranks}) achieving the optimal value $\mathrm{OPT}_D(I)$, and show that, using a geometric scaling argument, $(1-\frac{1}{d}) x^\ast$ is a valid solution for all rank-1 $\Lambda_f$-CG cuts. Letting $\mathrm{OPT}^{(1)}_{\Lambda_f}(I)$ denote the optimum over $P^{(1)}_{\Lambda_f}$, this implies $\mathrm{OPT}^{(1)}_{\Lambda_f}(I) \geq (1-\frac{1}{d}) \mathrm{OPT}_D(I)$ (because we are dealing with a maximization problem), which allows us to conclude $\mathrm{OPT}_D(I) \le (1+\varepsilon) \mathrm{OPT}^{(1)}_{\Lambda_f}(I)$.

We proceed by proving that $\left(1-\frac{1}{d}\right) x^\ast$ is a valid solution for all rank-1 $\Lambda_f$-CG cuts. The proof is split into two cases: when the cut has low capacity (i.e., $\tilde{b} \le d$), we argue that $x^\ast$ is already a valid solution for the cut (\Cref{lem:sos_degree_all_ranks}); scaling the solution down by the factor $\left(1-\frac{1}{d}\right) \in (0,1)$ cannot violate a cut with nonnegative right-hand side, because the left hand side is scaled by the same factor, so $\left(1-\frac{1}{d}\right) x^\ast$ satisfies it as well. When, instead, the cut has high capacity (i.e., $\tilde{b} > d$), we show that the geometric scaling is enough to guarantee that $\left(1-\frac{1}{d}\right) x^\ast$ satisfies the cut (\Cref{lem:scaling_all_ranks}).

\begin{lemma}[Constant $\sos$ Level for Low-Capacity Cuts]
\label{lem:sos_degree_all_ranks}
Let $I$ be an instance of a maximization problem subject to a system of generalized Horn constraints satisfying \Cref{ass:k_slack} for some constant $k$. Consider a non-dominated rank-1 $\Lambda_f$-CG cut $\sum_{i=1}^n c^{CG}_i x_i \le \lfloor \tilde{b}\rfloor$ generated with multipliers $\lambda_\ell \in \Lambda_f$ for some constant $f\ge 2$, with fractional capacity $\tilde{b} = \sum_\ell \lambda_\ell b_\ell \ge 0$, coefficients $c^{CG} \in \mathbb{Z}^n$, and negative support $K := \{i : c^{CG}_i < 0\}$. Then its violation threshold $\tau$ satisfies $\tau \le \bar\tau := (k+f+1)\lfloor \tilde{b} \rfloor + (k+f-1)$ and $|K| \le (k+f)\lfloor \tilde{b} \rfloor + (k+f-1)$, so by \Cref{lem:sos_exactness} the cut admits a $(\bar\tau + |K| + 1)$-$\sos$ certificate, where $\bar\tau + |K| + 1 \le \bigl(2(k+f)+1\bigr)\lfloor\tilde{b}\rfloor + \bigl(2(k+f) - 1\bigr)$.

In particular, if $\tilde{b} \le d$ for some integer $d \ge 2$ then, by \Cref{rem:duality}, the $D$-$\sos$ relaxation with $D = \bigl(2(k+f)+1\bigr)d + \bigl(2(k+f) - 1\bigr)$ satisfies the cut.
\end{lemma}

\begin{proof}
We can write the non-dominated cut in the mixed-sign form of \Cref{eq:generic_f}, with positive support $S = \{i : c^{CG}_i > 0\}$, negative support $K = \{i : c^{CG}_i < 0\}$, and right-hand side $\lfloor\tilde{b}\rfloor$: $\sum_{i\in S} c^{CG}_i x_i - \sum_{i\in K} (-c^{CG}_i) x_i \le \lfloor\tilde{b}\rfloor$.

By \Cref{thm:restricted_weight_bound}, the total negative weight satisfies $\sum_{i \in K} (-c^{CG}_i) \le (k+f)\lfloor\tilde{b}\rfloor + (k+f) - 1$. As every $-c^{CG}_i \ge 1$ for $i \in K$, this gives $|K| \le \sum_{i \in K}(-c^{CG}_i) \le (k+f)\lfloor\tilde{b}\rfloor + (k+f) - 1$.

Since the positive coefficients are integers $\ge 1$, the violation threshold $\tau$ (\Cref{def:violation_threshold}) of the constraint satisfies
\[
    \tau \;\le\; \tilde{b} + \sum_{i \in K}(-c^{CG}_i) \;\le\; \tilde{b} + (k+f)\lfloor \tilde{b}\rfloor + (k+f) - 1 ,
\]
and, since $\tau$ is an integer and $(k+f)\lfloor\tilde{b}\rfloor + (k+f) - 1 \in \mathbb{Z}$, $\tau \le \lfloor \tilde{b} + (k+f)\lfloor\tilde{b}\rfloor + (k+f) - 1\rfloor = (k+f+1)\lfloor\tilde{b}\rfloor + (k+f) - 1$.

The cut is a nonnegative combination of the constraints of $P$ (\Cref{rem:cuts_first_closure}), its violation threshold is at most $\bar\tau$, and $\bar\tau + |K| \ge \bar\tau \ge k+f-1 \ge 2 \ge 1$ since $k \ge 1$ and $f \ge 2$. Hence \Cref{lem:sos_exactness} applies with threshold bound $\bar\tau$ and the cut admits a $(\bar\tau + |K| + 1)$-$\sos$ certificate, where
$\bar\tau + |K| + 1 \le \bigl((k+f+1)\lfloor\tilde{b}\rfloor + (k+f-1)\bigr) + \bigl((k+f)\lfloor\tilde{b}\rfloor + (k+f-1)\bigr) + 1 = \bigl(2(k+f)+1\bigr)\lfloor\tilde{b}\rfloor + \bigl(2(k+f) - 1\bigr)$.

If $\tilde{b} \le d$ ($d\in \mathbb{Z}_{\ge 2}$), this is at most $D = \bigl(2(k+f)+1\bigr)d + \bigl(2(k+f) - 1\bigr)$, so the cut admits a $D$-$\sos$ certificate and, by \Cref{rem:duality}, the $D$-$\sos$ relaxation satisfies it.

\end{proof}

\begin{lemma}[Geometric Scaling for High-Capacity Cuts]
\label{lem:scaling_all_ranks}
Fix an integer $d \ge 2$ and let $\phi := 1 - \tfrac1d \in (0,1)$. Let $I$ be an instance of a maximization problem subject to a system of constraints satisfying \Cref{ass:k_slack}. Let $\sum_{i=1}^n c_i^{CG} x_i \le \lfloor \tilde{b}\rfloor$ be a rank-1 CG cut obtained by rounding a fractional inequality $\sum_{i=1}^n c_i^{CG} x_i \le \tilde{b}$ valid for the polytope $P$ underlying $I$, with $c^{CG} \in \mathbb{Z}^n$ and high capacity $\tilde{b} > d$. Then for every $y \in  P$ satisfying the fractional inequality $\sum_{i=1}^n c_i^{CG} y_i \le \tilde{b}$, it holds that $\phi\, y$ satisfies the rounded cut:
\[
    \sum_{i=1}^n c_i^{CG}\, \phi\, y_i \;\le\; \lfloor \tilde{b}\rfloor .
\]
\end{lemma}

\begin{proof}
Since $\phi = 1 - \tfrac1d \ge 0$, multiplying the fractional inequality
$\sum_i c_i^{CG} y_i \le \tilde{b}$ by $\phi$ yields
\[
    \sum_{i=1}^n c_i^{CG}\, \phi\, y_i = \phi \sum_{i=1}^n c_i^{CG} y_i
    \;\le\; \phi\, \tilde{b} = \tilde{b} - \frac{\tilde{b}}{d} .
\]
As $\tilde{b} > d$, we have $\tilde{b}/d > 1$, so $\tilde{b} - \tilde{b}/d < \tilde{b} - 1 < \lfloor \tilde{b}\rfloor$. Hence $\sum_{i=1}^n c_i^{CG}\, \phi\, y_i \le \lfloor \tilde{b}\rfloor$.
\end{proof}

\subsection{Approximation Guarantee}\label{sec:approx_guarantee}
We combine \Cref{lem:sos_degree_all_ranks,lem:scaling_all_ranks} to show that a scaled solution satisfies all rank-1 CG cuts.

\begin{corollary}[The Scaled $D$-$\sos$ Solution Lies in the First $\Lambda_f$-CG Closure]
\label{cor:scaled_in_Pt}
Let $I$ be an instance of a maximization problem subject to generalized Horn constraints satisfying \Cref{ass:k_slack} for some constant $k$; fix the integers $d \ge 2$ and $f \ge 2$, and let $x^\ast \in [0,1]^{n}$ be a feasible solution of the $D$-$\sos$ relaxation of $\mathcal{F}_{A,b}$, with $D=\bigl(2(k+f)+1\bigr)d+\bigl(2(k+f)-1\bigr)$. Then $\bigl(1-\tfrac1d\bigr) x^\ast \in P^{(1)}_{\Lambda_f}$.
\end{corollary}

\begin{proof}
Let $\phi := 1 - \tfrac1d \in (0,1)$ and $\hat{x} := \phi\, x^\ast$. $x^\ast$ satisfies the system's constraints and lies in $[0,1]^{n}$, so $x^\ast \in P^{(0)} = P$. Since $P^{(1)}_{\Lambda_f}$ is the intersection of its rank-1 $\Lambda_f$-cuts and every dominated cut is implied by a non-dominated one, it suffices to show that $\hat{x}$ satisfies every non-dominated rank-1 $\Lambda_f$-cut. Let
$\sum_{i=1}^n c_i^{CG} x_i \le \lfloor \tilde{b}\rfloor$ be such a cut, derived
from a fractional inequality $\sum_{i=1}^n c_i^{CG} x_i \le \tilde{b}$ valid for $P$ (with $c^{CG} \in \mathbb{Z}^n$ and $\tilde{b}\ge 0$, hence $\lfloor\tilde{b}\rfloor \ge 0$). We show $\hat{x}$ satisfies the cut.

\emph{Case $\tilde{b} \le d$.} By \Cref{lem:sos_degree_all_ranks} the cut admits a $t$-$\sos$ certificate with $t \le \bigl(2(k+f)+1\bigr)\lfloor\tilde{b}\rfloor + \bigl(2(k+f) - 1\bigr) \le \bigl(2(k+f)+1\bigr)d + \bigl(2(k+f) - 1\bigr)$, so by \Cref{rem:duality} the $D$-$\sos$ solution $x^\ast$ satisfies it, and $\sum_i c_i^{CG} x^\ast_i \le \lfloor\tilde{b}\rfloor$. Since $\lfloor\tilde{b}\rfloor \ge 0$ and $\phi \in (0,1)$, the scaling operation preserves the inequality:
\[
    \sum_{i=1}^n c_i^{CG} \hat{x}_i = \phi \sum_{i=1}^n c_i^{CG} x^\ast_i \le \lfloor\tilde{b}\rfloor
\]
(if $\sum_i c_i^{CG} x^\ast_i \ge 0$ the factor $\phi \le 1$ only lowers the left-hand side; otherwise if the lhs is negative, it is automatically $\le \lfloor\tilde{b}\rfloor$).

\emph{Case $\tilde{b} > d$.} Since $x^\ast \in P$ satisfies the fractional inequality $\sum_i c_i^{CG} x^\ast_i \le \tilde{b}$, by \Cref{lem:scaling_all_ranks}
\[
    \sum_{i=1}^n c_i^{CG} \hat{x}_i = \sum_{i=1}^n c_i^{CG}\, \phi\, x^\ast_i \le \lfloor\tilde{b}\rfloor .
\]
In both cases $\hat{x}$ satisfies the cut; as this holds for every non-dominated cut, $\hat{x}\in P^{(1)}_{\Lambda_f}$.
\end{proof}

We are finally ready to prove \Cref{thm:main_ptas}.

\begin{proof}[Proof of \Cref{thm:main_ptas}]

Let $d := \lceil 1/\ln(1+\varepsilon)\rceil + 1$ (notice that $d \ge 2$ for every $\varepsilon > 0$) and $D := (2(k+f)+1)d + (2(k+f)-1)$ as in \Cref{cor:scaled_in_Pt}; since $\ln(1+\varepsilon) \ge \varepsilon/2$ for $\varepsilon \in (0,1]$, we have $d = O(1/\varepsilon)$ and hence $D = O((k+f)/\varepsilon)$.

The lower bound $\mathrm{OPT}(I) \le \mathrm{OPT}_D(I)$ holds because $\mathrm{OPT}_D(I)$ is the optimum over a relaxation of the integer polytope for a maximization problem.

For the upper bound, let $x^\ast$ attain $\mathrm{OPT}_D(I) = c^\top x^{\ast}$.
By \Cref{cor:scaled_in_Pt}, $\hat{x} := (1-\tfrac1d) x^\ast \in P^{(1)}_{\Lambda_f}$, so $c^\top \hat{x} \le \mathrm{OPT}^{(1)}_{\Lambda_f}(I)$; by linearity $c^\top \hat{x} = (1-\tfrac1d) \mathrm{OPT}_D(I)$, hence $\bigl(1-\tfrac1d\bigr) \mathrm{OPT}_D(I) \le \mathrm{OPT}^{(1)}_{\Lambda_f}(I)$.
We can rewrite $-\ln(1-\tfrac1d) = \ln\bigl(1+\tfrac{1}{d-1}\bigr) \le \tfrac{1}{d-1}$ (by the standard inequality $\ln(1+x) \le x$ valid for all $x > -1$, applied with $x = \frac{1}{d-1}$), thus, recalling $d := \lceil 1/\ln(1+\varepsilon)\rceil + 1$:
\[
    \ln \bigl(1-\tfrac1d\bigr)^{-1} = -\ln\bigl(1-\tfrac1d\bigr) \le \frac{1}{d-1}
    = \frac{1}{\lceil 1/\ln(1+\varepsilon) \rceil} \le \ln(1+\varepsilon).
\]

Hence, $(1-\tfrac1d)^{-1} \le 1+\varepsilon$ (since $\ln$ is a strictly increasing function), and, since $\mathrm{OPT}^{(1)}_{\Lambda_f}(I) \ge 0$ (the all-zero vector lies in $P^{(1)}_{\Lambda_f}$, as every cut has nonnegative right-hand side),
\[
    \mathrm{OPT}_D(I) \;\le\; \bigl(1-\tfrac1d\bigr)^{-1}\,\mathrm{OPT}^{(1)}_{\Lambda_f}(I)
    \;\le\; (1+\varepsilon)\,\mathrm{OPT}^{(1)}_{\Lambda_f}(I).
\]
\end{proof}

\section{Discussion and Future Work}
\label{sec:discussion}

In this paper we proved that, using the $\sos$ hierarchy, it is possible to approximate to any fixed precision $\varepsilon$ the rank-1 $\Lambda_f$-CG closure for any constant $f \ge 2$.

It is clear that we did not answer \Cref{Q:PTAS_MIN} exhaustively: our algorithm handles only a subclass of min-closed systems, namely the systems of generalized Horn constraints satisfying \Cref{ass:k_slack}.

It is worth making explicit what this subclass is, in logical terms. By the Horn characterization of min-closed relations~\cite{JeavonsC95}, every min-closed Boolean feasible region is a conjunction of Boolean Horn clauses, and \Cref{ass:k_slack} separates these clauses according to their premise width. A \emph{composite} clause, of premise width $|P| \ge 2$, has the linear representation \eqref{eq:horn_linear} with capacity $b = |P| - 1 \ge 1$ and slack coefficient $q = 1$, so that $q \le b$: composite logic satisfies \Cref{ass:k_slack} with $k = 1$. A \emph{unitary} implication $x_i \Rightarrow y$, on the other hand, has the representation $x_i - y \le 0$, with $b = 0$ and $q = 1$, which violates $q \le k\,b$ for every constant $k$. Our regime requires the systems to be formulated with a positive capacity, and its boundary consists of the \emph{zero-capacity} constraints. We remark that a unitary implication can also be written with a positive capacity, for instance as $a x_i - q y \le b$ with $b \ge 1$, in which case it does satisfy \Cref{ass:k_slack}; it is the zero-capacity formulation that falls outside our analysis. This leaves the following question open.

\begin{question}[Approximability of Zero-Capacity Systems]
\label{Q:PTAS_zero_cap_open}
Let $\mathcal{F}_{A,b}$ be a min-closed Boolean feasible region whose linear formulation contains zero-capacity implications $x_i - y \le 0$. Does there exist a PTAS for maximizing linear objectives over its first Chv\'atal--Gomory closure? Does there exist one based on the $\sos$ hierarchy?
\end{question}

Other natural questions remain open as well: can methods other than the $\sos$ hierarchy achieve the same guarantees, and are there complexity lower bounds that prevent extending them to larger subclasses of min-closed systems?

\bibliography{lipics-v2021-sample-article}

\appendix

\section{Background: Sum of Squares Over the Boolean Hypercube}\label{sec:preliminaries}

We recall the formal definition of the $\sos$ proof system restricted to the Boolean domain (see e.g. \cite{FlemingKothariPitassi19} for further details). 

\begin{definition}[Boolean $\sos$ Refutation and Derivation]\label{def:sos_system}
Let $x = (x_1, \dots, x_n)$ be real variables. Let $\mathcal{I}_2 = \langle x_i^2 - x_i \mid i \in [n] \rangle$ denote the standard Boolean ideal. Two polynomials $A, B \in \mathbb{R}[x]$ are congruent modulo $\mathcal{I}_2$, written $A \equiv B \pmod{\mathcal{I}_2}$, if $A - B \in \mathcal{I}_2$.

Let $\mathcal{P} = \{P_j(x) \ge 0 \mid j \in [m]\}$ be a system of polynomial inequalities. A degree-$d$ $\sos$ derivation of a target inequality $C(x) \ge 0$ from $\mathcal{P}$ is an identity of the form:
\begin{equation}\label{eq:sos_def}
    C(x) = S_0(x) + \sum_{j=1}^m S_j(x) P_j(x) + \sum_{i=1}^n Q_i(x)(x_i^2 - x_i)
\end{equation}
where $S_0, S_1, \dots, S_m \in \mathbb{R}[x]$ are sums of squares of polynomials, $Q_1, \dots, Q_n \in \mathbb{R}[x]$ are arbitrary polynomial multipliers, and the degree of every term in the expansion is at most $d$.

A degree-$d$ $\sos$ refutation of $\mathcal{P}$ is a derivation of the constant polynomial $C(x) = -1$.

Following~\cite{MastrolilliHsos}, for an integer $t \ge 1$ a $t$-$\sos$ certificate of $C(x) \ge 0$ from $\mathcal{P}$ is a derivation \eqref{eq:sos_def} in which every polynomial squared in $S_0, S_1, \dots, S_m$ has degree at most $t$ (the multipliers $Q_i$ are unconstrained). Since the $P_j$ are linear, a $t$-$\sos$ certificate is in particular a degree-$(2t+1)$ derivation, and a $t$-$\sos$ certificate is also a $t'$-$\sos$ certificate for every $t' \ge t$.

\end{definition}

\begin{remark}[Multilinearization and Indicators]
Over the quotient ring $\mathbb{R}[x]/\mathcal{I}_2$, every polynomial reduces to a unique multilinear polynomial. Consequently, any nonnegative function over $\{0,1\}^n$ admits an exact $\sos$ representation. Specifically, for any point $v \in \{0,1\}^n$, its multilinear indicator $\chi_v(x) = \prod_{i: v_i=1} x_i \prod_{i: v_i=0} (1-x_i)$ is idempotent modulo $\mathcal{I}_2$, satisfying $\chi_v(x) \equiv \chi_v(x)^2 \pmod{\mathcal{I}_2}$. Thus, standard variable bounds such as $x_i \ge 0$ and $1-x_i \ge 0$ are certified at degree 2 via the identities $x_i \equiv x_i^2$ and $1-x_i \equiv (1-x_i)^2$.
\end{remark}

\begin{remark}[Degree Conventions in Boolean Proof Complexity]
In proof complexity (e.g.~\cite{FlemingKothariPitassi19,ODonnell17}) the \emph{degree} of a derivation is the total degree of the identity, as in \Cref{def:sos_system}. In the optimization literature, and in~\cite{MastrolilliHsos}, the hierarchy is indexed by the degree $t$ of the polynomials being squared: the $t$-$\sos$ relaxation is the semidefinite program whose moment matrix is indexed by the multilinear monomials of degree at most $t$, hence has size $n^{O(t)}$, and whose dual cone consists of the $t$-$\sos$ certificates. Throughout this paper the parameters $\tau+|K|+1$ (\Cref{lem:sos_exactness}) and $D$ (\Cref{thm:main_ptas}) are levels in the second sense.
\end{remark}

\begin{remark}[Duality Between Derivations and Relaxations]\label{rem:duality}
    The $t$-$\sos$ relaxation of a system $\mathcal{P}$ optimizes over \emph{pseudo-expectation operators}: linear functionals $\tilde{\mathbb{E}}$ on multilinear polynomials of degree at most $2t+1$ with $\tilde{\mathbb{E}}[1] = 1$, $\tilde{\mathbb{E}}[q^2] \ge 0$ and $\tilde{\mathbb{E}}[q^2 P_j] \ge 0$ for every $q$ of degree at most $t$ and every $j \in [m]$, products being reduced modulo $\mathcal{I}_2$. 
    Applying $\tilde{\mathbb{E}}$ to a $t$-$\sos$ certificate \eqref{eq:sos_def} of $C(x) \ge 0$ gives $\tilde{\mathbb{E}}[C] = \tilde{\mathbb{E}}[S_0] + \sum_j \tilde{\mathbb{E}}[S_j P_j] \ge 0$; since all polynomials are reduced modulo $\mathcal{I}_2$ before evaluation, the ideal terms vanish.
    Hence every solution of the $t$-$\sos$ relaxation satisfies every inequality admitting a $t$-$\sos$ certificate from $\mathcal{P}$~\cite[Proposition~2.2]{MastrolilliHsos}. 
    This is the direction used in \Cref{lem:sos_degree_all_ranks} and \Cref{cor:scaled_in_Pt}. We identify a solution $\tilde{\mathbb{E}}$ of the $t$-$\sos$ relaxation with its vector of first moments $x^\ast \in [0,1]^n$, $x^\ast_i := \tilde{\mathbb{E}}[x_i]$, and say that $x^\ast$ satisfies a linear inequality when $\tilde{\mathbb{E}}$ satisfies it.

\end{remark}

\section{The Exponential Expressivity Gap}\label{app:exponential_gap}

\begin{example}[The Exponential Expressivity Gap]\label{ex:exponential_gap}
To illustrate the expressive power of weighted min-closed constraints over classical Horn-$\SAT$, consider a system with $n$ prerequisite variables $x_1, \dots, x_n$ and a single target variable $y$. We want to enforce the logical condition: ``If any $\theta$ prerequisites are true, then the target $y$ must be true''.
Using a single weighted inequality, this is exactly captured by:
\begin{equation}\label{eq:threshold_min_closed}
    \sum_{i=1}^n x_i - n \cdot y \le \theta - 1
\end{equation}
This inequality is a generalized Horn constraint, so it is min-closed due to \Cref{prop:min_closed}. In standard Horn-$\SAT$, however, clauses are unweighted ($x_{i_1} \wedge \dots \wedge x_{i_c} \implies y$). To represent the exact same Boolean feasible region without auxiliary variables, a purely logical CNF formulation must explicitly forbid every possible subset of $\theta$ prerequisites from being true when $y=0$. This requires the conjunction of all possible minimal Horn clauses of premise width $\theta$:
\begin{equation}
    \bigwedge_{S \subseteq [n] : |S|=\theta} \left( \sum_{i \in S} x_i - y \le \theta - 1 \right)
\end{equation}
Setting the threshold to $\theta = \lfloor n/2 \rfloor$, the equivalent Horn-$\SAT$ formulation requires $\binom{n}{\lfloor n/2 \rfloor} = \Theta\left(\frac{2^n}{\sqrt{n}}\right)$ clauses. Also notice that for $\theta = \lfloor n/2 \rfloor$, the generalized Horn constraint also satisfies \Cref{ass:k_slack} for a small constant $\approx 2$. This exponential compression highlights the necessity of handling dense weights natively.
\end{example}

\section{Proof of \texorpdfstring{\Cref{lem:sos_exactness}}{SOS Exactness for Bounded Thresholds Lemma}}\label{sect:proof_lem:sos_exactness}
\begin{proof}[Proof of \Cref{lem:sos_exactness}]
Put $y'_k := 1 - y_k$ for $k \in K$, and let
\[
    z := (x_S, y'_K) \in \{0,1\}^N, \qquad w := (a, q) \in \mathbb{Z}^N_{\ge 1}, \qquad C := b + \sum_{k \in K} q_k \in \mathbb{R}_{\ge 0},
\]
where $N := |S| + |K|$.
It can be seen that $C - w^\top z = b - \sum_{i \in S} a_i x_i + \sum_{k \in K} q_k y_k$ as polynomials in $(x,y)$, and since $\sum_k q_k \in \mathbb{Z}$ it also holds that $\lfloor C \rfloor = \lfloor b \rfloor + \sum_k q_k$, so that the CG cut can be written as:
\begin{equation}\label{eq:target_packing_form}
    \lfloor b \rfloor - \sum_{i \in S} a_i x_i + \sum_{k \in K} q_k y_k = \lfloor C \rfloor - w^\top z.
\end{equation}
Hence, the target inequality is $\lfloor C \rfloor - w^\top z \ge 0$, written in the variables $z$.

Let $\mathcal{P}' := \{ z \in \{0,1\}^N : w^\top z \le C \}$; notice that $\mathcal{P}'$ is a packing instance with the single row $w^\top \in \mathbb{Z}^{1 \times N}_+$ and capacity $C \in \mathbb{R}_+$. For $z \in \mathcal{P}'$ the number $w^\top z$ is an integer not exceeding $C$, hence not exceeding $\lfloor C \rfloor$; thus $\lfloor C \rfloor - w^\top z \ge 0$ is a valid inequality for $\mathcal{P}'$, with $a_0 = \lfloor C \rfloor \ge 0$ and coefficient vector $w \ge 0$.

Now we prove that the violation threshold of the inequality $w^\top z \leq C$ is at most $\tau + |K|$. Sort $w_1 \le \dots \le w_N$ and let $\tau^\ast := \max\{ i : \sum_{j \le i} w_j \le C \}$ be the violation threshold of $w^\top z \le C$. If $\tau + |K| \ge N$, we are done since $\tau^\ast \leq N$. Else, take any $i$ such that $N \ge i > \tau + |K|$. Among the $i$ smallest entries of $w$ at most $|K|$ come from the block $q$, so at least $i - |K| > \tau$ variables come from the block $a$; therefore $\sum_{j \le i} w_j$ is at least the sum of the $\tau + 1$ smallest entries of $a$, which exceeds $b + \sum_{k \in K} q_k = C$ by the definition of $\tau$. Hence $\tau^\ast \le \tau + |K|$. Notice that the violation threshold for the inequality $w^\top z \le \lfloor C \rfloor$ is the same, since all $w$'s are integral.

We can apply \Cref{lem:decomposition_theorem} to $\mathcal{P}'$ with $\pi:=  \tau + |K| \ge 1$ by hypothesis, and $\pi \ge \tau^\ast$ by the previous argument: hence, there exist sums of squares $\sigma_0, \sigma_1$ of polynomials of degree at most $\pi + 1$ in $z$, and polynomials $Q_1, \dots, Q_N$, such that
\begin{equation}\label{eq:certificate_packing}
    \lfloor C \rfloor - w^\top z = \sigma_0(z) + \sigma_1(z)\,\bigl(C - w^\top z\bigr) + \sum_{j=1}^N Q_j(z)\,(z_j^2 - z_j).
\end{equation}

By hypothesis the fractional inequality is a nonnegative combination of the constraints $g_\ell(x,y) \ge 0$ of $P$: i.e., there are $\lambda_\ell \ge 0$ with
\[
    C - w^\top z = b - \sum_{i \in S} a_i x_i + \sum_{k \in K} q_k y_k = \sum_{\ell} \lambda_\ell\, g_\ell(x,y).
\]
Substitute $z = (x_S, 1 - y_K)$ throughout \eqref{eq:certificate_packing}. The substitution is an invertible affine change of variables, so it preserves degrees and maps squares to squares; hence $\sigma_0(x_S, 1 - y_K)$ and each $\lambda_\ell\, \sigma_1(x_S, 1 - y_K)$ are sums of squares of polynomials of degree at most $\pi + 1$, and
\[
    \sigma_1(z)\,\bigl(C - w^\top z\bigr) = \sum_{\ell} \bigl[\lambda_\ell\, \sigma_1(x_S, 1 - y_K)\bigr]\, g_\ell(x,y).
\]
The ideal terms remain in the Boolean ideal even after the substitution, since $(1 - y_k)^2 - (1 - y_k) = y_k^2 - y_k$.

Finally, we can write:
\begin{align*}
    \lfloor b \rfloor - \sum_{i \in S} a_i x_i + \sum_{k \in K} q_k y_k
    &= \sigma_0(x_S, 1 - y_K) + \sum_{\ell} \bigl[\lambda_\ell\, \sigma_1(x_S, 1 - y_K)\bigr]\, g_\ell(x,y) \\
    &\quad + \sum_{i \in S} \tilde{Q}_i(x,y)\,(x_i^2 - x_i) + \sum_{k \in K} \tilde{Q}_k(x,y)\,(y_k^2 - y_k),
\end{align*}
for suitable polynomials $\tilde{Q}$ obtained by substituting $z = (x_S, 1-y_K)$ in $Q_j$'s. The latter is a $(\pi + 1)$-$\sos$ certificate, that is, a $(\tau + |K| + 1)$-$\sos$ certificate of the Chv\'atal--Gomory rounding from the constraints of $P$.
\end{proof}

\section{Proof of \texorpdfstring{\Cref{thm:cg-bound-redundancy}}{Theorem Redundancy of Non-Minimal Multipliers for Non-Negativity Bounds}}\label{app:proof_of_thm_cg-bound-redundancy}

\begin{proof}[Proof of \Cref{thm:cg-bound-redundancy}]
    
Let $\beta_{frac} = \sum_{\ell=1}^m \lambda_\ell b_\ell$, and let the rounded capacity be $\beta = \lfloor \beta_{frac} \rfloor$. By applying the structural multipliers $\lambda \in \Lambda_f^m$ to the initial system and the bound multipliers $\mu_i \ge 0$ to $-x_i \le 0$, the generic rank-1 $\Lambda_f$-CG cut $C_{generic}$ is obtained by rounding down the right-hand side of the valid conic combination:

\[
    \sum_{i=1}^n \big( w_i(\lambda) - \mu_i \big) x_i \le \beta_{frac}.
\] 

For the CG rounding operation to be valid, each coefficient $(w_i(\lambda) - \mu_i)$ must belong to $\mathbb{Z}$. Therefore, any valid multiplier $\mu_i$ can be uniquely decomposed as $\mu_i = \mu^{tight}_i + d_i$, where $\mu^{tight}_i \in [0, 1)$ absorbs the fractional part of $w_i(\lambda)$, and $d_i \in \mathbb{Z}_{\ge 0}$ is a nonnegative integer. 

We define the tight cut $C_{tight}$ by selecting the minimal fractional absorber $\mu_i = \mu^{tight}_i$ (which sets $d_i = 0$ for all variables):
\[
    \sum_{i=1}^n \big( w_i(\lambda) - \mu^{tight}_i \big) x_i \le \beta.
\]

Notice that both $C_{tight}$ and $C_{generic}$ share the same right hand side, since the inequalities $-x_i \leq 0$ contribute nothing to it.

Consider an arbitrary point $x \in \mathbb{R}_{\ge 0}^n$. We express the left-hand side (LHS) of $C_{generic}$ relative to the LHS of $C_{tight}$:
\begin{align}
    \text{LHS}_{generic}(x) & = \sum_{i=1}^n \big( w_i(\lambda) - \mu^{tight}_i - d_i \big) x_i \\
    & = \text{LHS}_{tight}(x) - \sum_{i=1}^n d_i x_i.
\end{align}
Because the point belongs to the nonnegative orthant ($x_i \ge 0$) and the integer offsets are nonnegative ($d_i \ge 0$), the subtraction term is nonnegative ($\sum d_i x_i \ge 0$). This establishes for all $x \ge 0$:
\[
    \text{LHS}_{generic}(x) \le \text{LHS}_{tight}(x)
\]
Assume $x$ satisfies the tight cut, meaning $\text{LHS}_{tight}(x) \le \beta$. By the inequality above, it follows that:
\[
    \text{LHS}_{generic}(x) \le \text{LHS}_{tight}(x) \le \beta.
\]
Thus, any point that satisfies $C_{tight}$ satisfies $C_{generic}$. Since $P^{(1)}_{\Lambda_f}$ is the intersection of $P$ with all valid rank-1 $\Lambda_f$-cuts, incorporating the bound $-x_i \le 0$ with any integer weight $d_i \ge 1$ beyond the minimal fractional absorber $\mu^{tight}_i$ is redundant.
\end{proof}

\section{Proof of \texorpdfstring{\Cref{thm:cg-bound-preservation}}{Theorem Preservation of Negative Weight Bounds in the CG Closure}}\label{app:proof-cg-bound-preservation}
\begin{proof}[Proof of \cref{thm:cg-bound-preservation}]
Let $\lambda \in \Lambda_f^m$ be the structural multipliers generating $C_{tight}$, with fractional capacity $\beta_{frac} = \sum_{\ell=1}^m \lambda_\ell b_\ell$ and rounded capacity $\beta = \lfloor \beta_{frac} \rfloor$. By definition of the floor function:
\[
    \beta_{frac} < \beta + 1 .
\]
Consider the coefficient $w_i(\lambda)$ prior to bound absorption. Since the packing coefficients and the structural multipliers are nonnegative ($a_{\ell, i} \ge 0$ and $\lambda_\ell\ge 0$), dropping the positive terms yields the lower bound:
\[
    w_i(\lambda) \ge - \sum_{\ell : j_\ell = i} \lambda_\ell c_\ell .
\]
By \Cref{ass:k_slack}, $c_\ell \le k \cdot b_\ell$ for all $\ell$. This yields:
\[
    \sum_{\ell : j_\ell = i} \lambda_\ell c_\ell \le k \sum_{\ell : j_\ell = i} \lambda_\ell b_\ell \le k \sum_{\ell=1}^m \lambda_\ell b_\ell = k \beta_{frac} .
\]
Therefore, $-w_i(\lambda) \le k \beta_{frac}$.

For a non-dominated cut, \cref{thm:cg-bound-redundancy} implies that the non-negativity multipliers are minimal, meaning $\mu^{tight}_i \in [0, 1)$, for every $i\in[n]$. The integer coefficient of the CG cut is $c^{CG}_i = w_i(\lambda) - \mu^{tight}_i$. Assuming $c^{CG}_i < 0$, its magnitude satisfies:
\[
    |c^{CG}_i| = -c^{CG}_i = -w_i(\lambda) + \mu^{tight}_i \le k\cdot \beta_{frac} + \mu^{tight}_i .
\]
Substituting $\beta_{frac} < \beta + 1$ and $\mu^{tight}_i < 1$, we obtain the strict inequality:
\[
    |c^{CG}_i| < k(\beta + 1) + 1 .
\]
Since the Chv\'atal--Gomory procedure guarantees $c^{CG}_i \in \mathbb{Z}$, its absolute value $|c^{CG}_i|$ is an integer. For integers, the relation $< k(\beta + 1) + 1$ is equivalent to $\le k(\beta + 1)$. Hence:
\[
    |c^{CG}_i| \le k(\beta + 1) .
\]
This confirms that the magnitude of any negative coefficient of any non-dominated cut in $P^{(1)}_{\Lambda_f}$ with capacity $\beta$ is bounded by $O(\beta)$, when $k$ is fixed.
\end{proof}

\section{Proof of \texorpdfstring{\Cref{thm:restricted_weight_bound}}{Theorem Bound on Total Negative Weight}}\label{app:proof-restricted-weight-bound}
\begin{proof}[Proof of \cref{thm:restricted_weight_bound}]
Let $I_{active} := \{ \ell : \lambda_\ell > 0 \; \text{and} \; c_\ell \ge 1 \}$ be the set of active constraints with nonzero negative coefficient. These are the only constraints that affect the quantity $\Delta_{cut}$. If $I_{active} = \emptyset$, then $\Delta_{cut} = 0$, and the claim is trivially true. Hence, we will assume that $I_{active} \ne \emptyset$.

Since the cut is non-dominated, \Cref{thm:cg-bound-redundancy} forces the multiplier of each $-x_i \le 0$ to be minimal ($\mu^{tight}_i < 1$), so the final coefficient is $c^{CG}_i = \lfloor w_i(\lambda)\rfloor$ and, when negative, $|c^{CG}_i| = \lceil -w_i(\lambda)\rceil$. Dropping the nonnegative terms of $w_i(\lambda)$, using the fact that $\lceil\cdot\rceil$ is monotone and subadditive, and knowing that each constraint has a single penalty variable $j_\ell$, the total negative weight is bounded by a sum over the active constraints:
\[
    \Delta_{cut} \le \sum_{\ell \in I_{active}} \lceil \lambda_\ell c_\ell \rceil
\]
Using the inequality $\lceil x \rceil < x + 1$ for $x \ge 0$, and that $I_{active} \ne \emptyset$:
\[
    \Delta_{cut} < \sum_{\ell \in I_{active}} (\lambda_\ell c_\ell + 1) = \sum_{\ell \in I_{active}} \lambda_\ell c_\ell + |I_{active}|.
\]
Given \Cref{ass:k_slack}, we know that $c_\ell \le k b_\ell$ for every $\ell \in [m]$, and since $b_\ell \ge 0$ and $\lambda_\ell \ge 0$, we have:
\[
    \sum_{\ell \in I_{active}} \lambda_\ell c_\ell \le k \sum_{\ell \in I_{active}} \lambda_\ell b_\ell \le k \sum_{\ell = 1}^{m} \lambda_\ell b_\ell < k( \beta + 1).
\]
For the last term, we utilize the fact that for every $x\ge 0$ it holds $x < \lfloor x \rfloor  + 1$. Since $\lambda_\ell \ge 1/f$ for the active constraints, it follows that $1 \le f \lambda_\ell$, hence $|I_{active}| = \sum_{\ell \in I_{active}} 1 \le \sum_{\ell \in I_{active}} f \lambda_\ell$.
Combining this with $b_\ell \ge 1$ for $\ell\in I_{active}$ (if $c_\ell \ge 1$, due to \Cref{ass:k_slack}, $b_\ell$ has to be positive):
\[
    |I_{active}| \le f \sum_{\ell \in I_{active}} \lambda_\ell b_\ell \le f \sum_{\ell \in [m]} \lambda_\ell b_\ell < f(\beta + 1).
\]
Substituting these bounds into the expression for $\Delta_{cut}$:
\[
    \Delta_{cut} < k(\beta + 1) + f(\beta + 1) = \beta(k + f) + k + f.
\]
Since $\Delta_{cut}$ must be an integer, it follows that $\Delta_{cut} \le \beta(k + f) + k + f - 1$.
\end{proof}

\end{document}